%% file: main.tex
\documentclass{ieeetj}
\usepackage[T1]{fontenc}
\usepackage{cite}
\usepackage{amsmath,amssymb,amsfonts}
\usepackage{graphicx,xcolor,booktabs,tabularx,array}
\usepackage{tikz}
\usetikzlibrary{arrows.meta,positioning,fit,backgrounds,calc}
\usepackage{float,etoolbox,microtype,balance}
\usepackage{hyperref}
\hypersetup{hidelinks,pdftitle={BR-MoE: Budgeted Semantic Release for Privacy-Preserving Traffic Analytics},pdfauthor={Abdolazim Rezaei}}
\AtBeginDocument{\definecolor{tmlcncolor}{cmyk}{0.93,0.59,0.15,0.02}\definecolor{NavyBlue}{RGB}{0,86,125}}
\definecolor{teal}{RGB}{0,108,121}
\definecolor{pale}{RGB}{236,245,248}
\definecolor{amber}{RGB}{157,89,0}
\newcommand{\method}{BR-MoE}
\newcommand{\illustrative}{\textbf{Illustrative values only. No experiments were run.}}
\newcolumntype{L}[1]{>{\raggedright\arraybackslash}p{#1}}
\newcolumntype{Y}{>{\raggedright\arraybackslash}X}
\floatstyle{ruled}
\newfloat{algorithm}{tbp}{loa}
\floatname{algorithm}{Algorithm}
\newtheorem{proposition}{Proposition}
\def\OJlogo{\small\sffamily RESEARCH MANUSCRIPT}
\def\seclogo{\small\sffamily BR-MoE / RESEARCH DRAFT}
\makeatletter
\let\@revised\@empty
\patchcmd{\@affil}{\textsuperscript{\the\tempcount}}{}{}{}
\patchcmd{\@maketitle}{\fi.\space\ifx\@editor}{\fi\space\ifx\@editor}{}{}
\def\ps@draftpaper{%
  \def\@oddhead{\small\sffamily\color{NavyBlue}Rezaei: Budgeted Semantic Release\hfill Research draft}%
  \let\@evenhead\@oddhead
  \let\@evenfoot\@oddfoot}
\let\ps@plain\ps@draftpaper
\let\ps@headings\ps@draftpaper
\makeatother
\begin{document}
\title{BR-MoE: Budgeted Semantic Release for Privacy-Preserving Traffic Analytics}
\author{Abdolazim Rezaei}
\affil{Department of Computer Science, Texas A\&M University--Corpus Christi, Corpus Christi, TX 78412, USA}
\corresp{Corresponding author: Abdolazim Rezaei (email: arezaei@islander.tamucc.edu).}
%\authornote{Research draft. The method and evaluation protocol are proposed. Performance values are illustrative and must be replaced with measured results before submission.}
\begin{abstract}
Traffic cameras support useful analysis, but their outputs can reveal private information. Replacing images with text reduces direct visual exposure. It does not place a formal bound on disclosure. The published reinforcement learning and mixture-of-experts framework refines scene descriptions through relevance, coverage, and brevity rewards. It does not account for repeated releases under a fixed privacy budget. This paper proposes budgeted semantic release as an extension of that framework. Four scene experts first process isolated video chunks. Their outputs are mapped to a fixed, bounded semantic schema. A reinforcement learning controller selects which expert summaries to release and how much privacy budget to spend. The controller uses public task requests and previous private releases. A separate privacy filter enforces the total budget. Calibrated noise is added before the summaries are converted to text. Under explicit chunk isolation and bounded-event assumptions, the complete release transcript satisfies an event-level differential privacy bound. The design also supports reuse of previous summaries without new disclosure. An evaluation protocol covers traffic utility, privacy attacks, release freshness, and computational cost. Illustrative tables show the intended reporting format. They do not establish empirical gains. The proposed contribution is a task-aware release interface with auditable privacy accounting for repeated traffic descriptions.
\end{abstract}
\begin{IEEEkeywords}
Differential privacy, intelligent transportation systems, mixture of experts, reinforcement learning, video analytics.
\end{IEEEkeywords}
\maketitle
\thispagestyle{draftpaper}
\input{sections/01_introduction}
\input{sections/02_related_work}
\input{sections/03_problem}
\input{sections/04_method}
\input{sections/05_evaluation}
\input{sections/06_results}
\input{sections/07_discussion}
\bibliographystyle{ieeetr}
\bibliography{references}
\end{document}

%% file: sections/01_introduction.tex
\section{Introduction}
\label{sec:intro}
\IEEEPARstart{T}{raffic} cameras support road monitoring, pedestrian analysis, and transport planning. Their images also contain faces, vehicle plates, and contextual details. A system can expose these details even when its main task does not require identity. Vision-language models can infer private attributes from indirect visual cues~\cite{tomekce2024private, rezaei_2026_23667820}. Privacy protection therefore needs to address the released information, not only visible faces.

The published RL-MoE method combines reinforcement learning (RL) with a mixture of experts (MoE)~\cite{rezaei2025rl}. Four experts describe traffic, pedestrians, road signs, and the environment. An image-conditioned network weights their outputs. An RL policy then refines the text. This architecture offers a useful basis for semantic scene analysis. However, its stated reward contains relevance, conciseness, and coverage terms. None directly defines a privacy budget. Its privacy evidence comes from attacks, and its limitations section explicitly identifies the lack of a formal guarantee.

Repeated descriptions create a second concern. Several individually coarse reports may expose a detailed event when combined. A privacy mechanism must therefore account for the full release transcript. Training a model with differential privacy (DP) does not solve this issue by itself. Training privacy protects contributions to the learned model. Release privacy protects the new data used to answer a query~\cite{dwork2014foundations,sander2024dpcap}.

This paper proposes \emph{Budgeted Release MoE}, abbreviated as \method. It extends RL-MoE with a bounded semantic interface and an explicit release controller. The four experts remain, but unrestricted captions no longer cross the release boundary. Each expert returns a small vector with a fixed range. The controller selects expert summaries and privacy budgets. Its state contains public requests and previously privatized outputs. A separate filter rejects actions that exceed the budget. A controlled renderer converts noisy summaries into text.

The work makes three proposed contributions. First, it defines a release interface that binds expert outputs to a fixed schema and a bounded-event privacy model. Second, it designs a task-aware RL controller whose decisions do not inspect fresh private features. Third, it specifies a matched evaluation of utility, repeated disclosure, freshness, and cost. The privacy argument builds on established DP mechanisms. The research question is whether learned release control improves utility over strong fixed and task-aware schedules at the same budget.

\textbf{Draft status.} The architecture and proof are specified here. No new model training or benchmark experiment has been completed for this manuscript. The numerical performance tables are labeled examples for later replacement. They do not support superiority claims.

%% file: sections/02_related_work.tex
\section{Related Work and Research Gap}
\label{sec:related}
\subsection{Visual Information Converted to Text}
Feedback-based vision-to-text refinement uses RL to improve scene descriptions~\cite{rezaei2025privacy}. RL-MoE adds four scene experts and weighted refinement~\cite{rezaei2025rl}. Video-to-Text Pedestrian Monitoring generates reports for intersection activity and safety analysis~\cite{abdelrahman2025vtpm}. EdgeVLM studies local captions for activity recognition and evaluates privacy perceptions~\cite{hiraoka2026edgevlm}. These works motivate semantic release. A change of modality alone does not establish a DP bound.

Synthesis via Private Textual Intermediaries (SPTI) uses private selection of textual intermediaries to generate a synthetic image dataset~\cite{wang2025spti}. Its privacy unit is a dataset entry. Its objective differs from answering a sequence of requests about a live traffic stream. DP-Cap instead protects image--caption training pairs during representation learning~\cite{sander2024dpcap}. It does not privatize the content of a new image supplied at inference. These distinctions prevent a misleading comparison between training privacy, dataset synthesis, and online release privacy.

Recent defenses also target direct information disclosure. Defeating Cerberus modifies internal model states to reduce responses to requests for personally identifiable information~\cite{zhang2026cerberus}. Such refusal behavior is useful, but it addresses a different question from cumulative disclosure through permitted traffic summaries.

\subsection{Private Video Analytics and Adaptive Budgets}
VideoDP releases sanitized video for later analysis~\cite{wang2020videodp}. Privid introduces event-duration privacy and bounds the influence of isolated video chunks~\cite{cangialosi2022privid}. This is the main foundation for the bounded-chunk design used here. GeoPrivd adds trajectory-level sensitivity control and a query language for geospatial traffic analytics~\cite{xu2026geoprivd}. Its slicing, bounded contributions, and query validation further narrow the novelty claim. The proposed method studies RL control over semantic-report releases. It does not claim a new definition of private video analytics or stronger privacy than these systems.

Privacy-budget allocation is also established research. Chen et al. use multi-agent RL to allocate budgets during federated learning~\cite{chen2024adaptive}. MTSP-LDP supports multiple streaming tasks under local DP~\cite{liu2025mtsp}. Privacy filters address adaptive choices of privacy parameters~\cite{rogers2016filters}. Recent work further studies adaptive budgeting~\cite{liang2026adaptive}. These methods show that adaptive allocation is not new on its own. The open design issue for this extension is how to combine task selection, semantic experts, reuse, and budget enforcement without exposing private routing signals.

\subsection{Broader AI and Transportation Context}
MINT-V2X provides a basis for studying coupled mobility and network behavior~\cite{anjum2026mint}. PEHT studies efficient traffic prediction with urban context~\cite{rezaei2026parameter}. GraphTrafficGPT organizes traffic-analysis tasks through a graph of agents~\cite{taleb2025graphtrafficgptenhancingtrafficmanagement}. These studies motivate downstream demand for compact scene information. PEFT-DML provides context for efficient multimodal perception~\cite{rezaei2026peft}. It is not a privacy mechanism. Work on imbalanced credit-card fraud detection illustrates a separate AI evaluation problem~\cite{rezaei2024credit}. It motivates reporting class-sensitive metrics for rare events here, but it provides no evidence about traffic privacy.

\begin{table*}[t]
\caption{Literature-based position of the proposed extension. Different privacy units are not interchangeable.}
\label{tab:position}
\centering\small
\begin{tabularx}{\textwidth}{L{2.4cm}YYY}
\toprule
Work & Main output or objective & Privacy scope & Relation to this study\\
\midrule
RL-MoE~\cite{rezaei2025rl} & Refined scene descriptions & Empirical privacy evaluation & Source architecture; lacks a release accountant.\\
VTPM~\cite{abdelrahman2025vtpm} & Pedestrian activity reports & Identity-related information removal & Relevant semantic-report baseline.\\
DP-Cap~\cite{sander2024dpcap} & Learned visual representations & Training image--caption records & Does not sanitize new inference inputs.\\
SPTI~\cite{wang2025spti} & Synthetic images via text & Entries in a private dataset & Batch synthesis; not a live-caption benchmark.\\
Privid~\cite{cangialosi2022privid} & Aggregate video query answers & Bounded-duration events & Foundation for isolated chunk processing.\\
GeoPrivd~\cite{xu2026geoprivd} & Geospatial traffic query answers & Trajectory-level influence control & Closely related analytics system; different release objective.\\
MTSP-LDP~\cite{liu2025mtsp} & Multi-task stream statistics & Window-based local DP & Related scheduling and reuse ideas.\\
\method\ (proposed) & Selected summaries and grounded text & At most $K$ changed chunks per protected event & Learns release choices under an enforced transcript budget.\\
\bottomrule
\end{tabularx}
\end{table*}

\subsection{Specific Gap and Testable Claim}
Table~\ref{tab:position} separates the relevant objectives. The gap is not the absence of DP, RL, or image-to-text conversion. It is the missing combination in RL-MoE of a bounded semantic query interface, privacy-safe expert selection, and cumulative release control. \method\ treats this combination as a proposed systems contribution. Its empirical value remains a hypothesis. A positive result requires an improvement over a task-aware private scheduler at matched privacy scope and budget.

%% file: sections/03_problem.tex
\section{Problem Formulation and Threat Model}
\label{sec:problem}
\subsection{Inputs, Tasks, and Released Information}
Let $V=(V_1,\ldots,V_N)$ be a video divided into nonoverlapping chunks of duration $c$. Chunk boundaries and reporting times are public. A request at release slot $t$ has a public task label $q_t$. Examples include vehicle occupancy, pedestrian activity, road-sign presence, and weather. The output $Y_t$ contains a textual summary, its reporting interval, and a freshness label. The full transcript is $\mathcal{Y}=(Y_1,\ldots,Y_T)$. It also includes observable release choices and budget metadata.

The target is aggregate monitoring. It is not identification, individual trajectory recovery, or immediate collision avoidance. A small release budget can make a report uninformative. The system must represent this outcome as uncertainty or unavailability.

\subsection{Adversary and Trust Boundary}
The observer can collect all released messages. The observer may issue adaptive public requests and combine the answers with external information. The observer knows the architecture, trained controller, schema, and privacy parameters. The camera processor and privacy filter are trusted. Raw frames, expert drafts, private caches, and noise seeds remain inside this boundary. The observer cannot inspect their execution state.

Expert processing must be isolated by chunk. Model weights remain fixed during deployment. Trackers, recurrent states, and retrieval caches are reset at chunk boundaries. Prompts and retrieval documents are fixed public resources. A raw frame must not enter a remote model service outside the trusted boundary. Report sizes and send times follow a public schedule. Failure details and private processing times are not released. These conditions are part of the model, not optional implementation advice.

\subsection{Protected Unit and Adjacency}
Two videos $V\sim_K V'$ have the same public timeline and differ in at most $K$ chunks. Changed chunks may contain arbitrary frame differences. This definition can cover a person's appearance, an identifying object, or another bounded event. It does not require a perfect identity detector. For every measurable transcript set $\mathcal{S}$, the target is
\begin{equation}
 \Pr[\mathcal{M}(V)\in\mathcal{S}]
 \leq e^{\varepsilon_\star}
 \Pr[\mathcal{M}(V')\in\mathcal{S}],
 \label{eq:dp}
\end{equation}
where $\mathcal{M}$ denotes the complete mechanism and $\varepsilon_\star$ is the total privacy budget. This is pure DP for the stated adjacency. Smaller budgets give a stronger bound.

An event of duration at most $\rho$ may straddle a chunk boundary. A conservative choice is $K=1+\lceil\rho/c\rceil$ for one appearance~\cite{cangialosi2022privid}. For $r$ separate appearances, use at least $r(1+\lceil\rho/c\rceil)$ unless a tighter bound is justified. A two-second event therefore needs $K=2$ when $c=2$ seconds. It must not be described as a four-second guarantee. The event bound applies across the declared observation horizon.

This is an event-level guarantee. It is not an unlimited person-level guarantee. Longer visits, repeated appearances, and multiple cameras require larger bounds or additional composition. User-level continual privacy is a distinct problem~\cite{dong2023continual}. Restarting an episode or a server does not erase earlier disclosure.

%% file: sections/04_method.tex
\section{Proposed Method}
\label{sec:method}
Figure~\ref{fig:architecture} shows the release boundary. RL-MoE supplies the four expert roles. \method\ changes their output interface and the role of RL. The policy controls disclosure across report slots. It does not edit unrestricted private text after disclosure.

\begin{figure*}[t]
\centering
\input{figures/architecture_tikz}
\caption{Proposed \method\ architecture. Blue solid arrows carry data. Dashed arrows carry control signals. All fresh private information reaches the output through the noise mechanism. The controller receives only public inputs and prior private releases. The offline trainer uses public or synthetic episodes and is frozen before deployment.}
\label{fig:architecture}
\end{figure*}
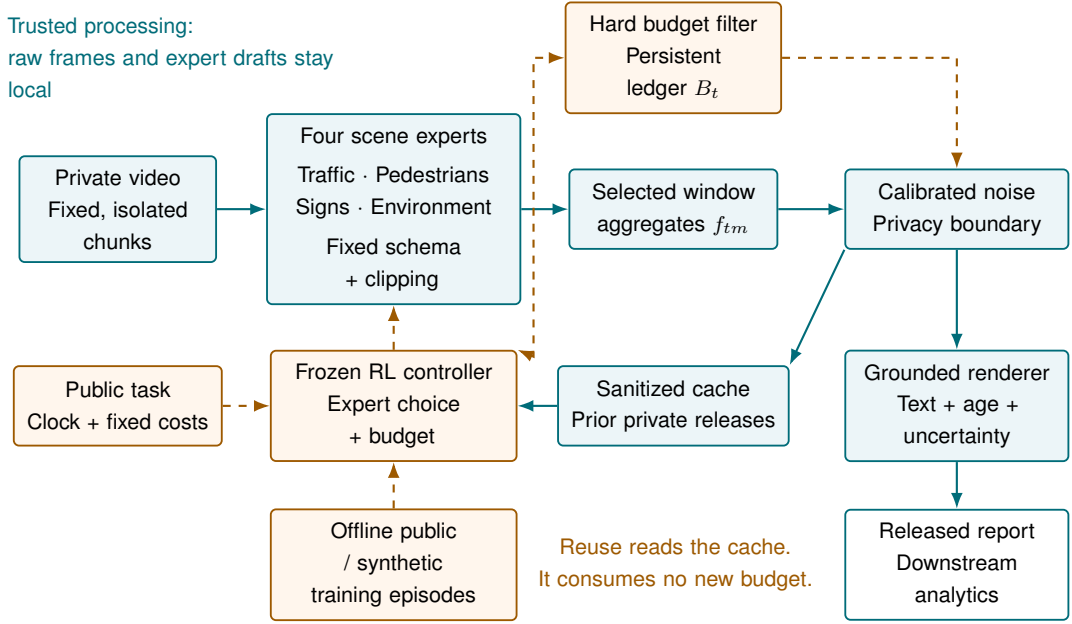

\subsection{Bounded Semantic Experts}
For chunk $i$ and expert $m\in\{1,2,3,4\}$, let $h_m(V_i)$ be a vector over a fixed public vocabulary. Negative entries are removed and the vector is clipped:
\begin{equation}
 v_{im}=\frac{[h_m(V_i)]_+}
 {\max\{1,\|[h_m(V_i)]_+\|_1/C_m\}}.
 \label{eq:clip}
\end{equation}
Here $[\cdot]_+$ applies the positive part to each component. The public cap $C_m>0$ ensures $\|v_{im}\|_1\leq C_m$. The boundary enforces this cap even if the expert fails. Invalid values, missing fields, or timeouts map to a fixed zero vector. Failure flags are not sent separately.

Table~\ref{tab:schema} gives the initial schema. It has twelve components in total. Vehicle and pedestrian values describe average occupancy within a chunk. They do not count unique identities or total arrivals. Sign and weather vectors contain class proportions over fixed sampled frames. No plate string, face embedding, clothing description, exact trajectory, or open text field is accepted. The absence of these fields improves data minimization. The formal bound follows from clipping and noise, not from this exclusion alone.

\begin{table}[t]
\caption{Proposed semantic schema and starting caps. Values are design choices, not tuned settings.}
\label{tab:schema}
\centering\small
\begin{tabularx}{\columnwidth}{L{1.6cm}Ycc}
\toprule
Expert & Components & $d_m$ & $C_m$\\
\midrule
Traffic & Moving; stopped vehicles & 2 & 20\\
Pedestrian & Walking; waiting pedestrians & 2 & 6\\
Signs & Stop; speed-limit; other sign; no sign & 4 & 1\\
Environment & Clear; rain; fog; other weather & 4 & 1\\
\bottomrule
\end{tabularx}
\end{table}

At report slot $t$, the public interval $\mathcal{I}_t$ contains $J_t$ chunks. The internal expert query is
\begin{equation}
 f_{tm}(V)=\frac{1}{J_t}\sum_{i\in\mathcal{I}_t}v_{im}.
 \label{eq:query}
\end{equation}
The denominator uses the public interval length. It is not a private count of visible objects. The first implementation uses fixed one-minute windows. It can reuse cached chunk vectors within the trusted boundary. Any future overlapping-window implementation must keep the same accounting discipline.

\subsection{Task-Aware Release Controller}
The controller state is $s_t=(q_t,\widetilde H_{t-1},B_t,t/T)$. The public task vector $q_t$ has four nonnegative weights that sum to one. The cache $\widetilde H_{t-1}$ stores each expert's most recent noisy vector, availability mask, and age. The value $B_t$ is the remaining budget. Missing estimates have a fixed public encoding. Fresh raw features, unnoised confidence scores, and private expert weights are excluded from the state.

The action $a_t$ selects a release amount $e_t$ and an allocation vector $w_t$. For an active expert, $\varepsilon_{tm}=e_t w_{tm}$, with $w_{tm}>0$ and $\sum_m w_{tm}=1$. Zero weight means that the expert is not queried. The finite action library contains four single-expert allocations, two equally weighted pairs (traffic--pedestrian and signs--environment), one uniform four-expert allocation, and one public task-weighted allocation. The latter uses $w_{tm}=(q_{tm}+0.05)/1.2$. Four positive choices of $e_t$ give 32 query actions. One reuse action gives 33 actions in total.

A two-layer actor samples an action. A separate mask removes any action with $e_t>B_t$. The reuse action is always available. It returns cached estimates with their original interval and current age. If no cache exists, the output is a fixed unavailable message. This prevents a hidden query when the budget is exhausted.

The policy uses a finite summary of history. Thus it is an approximate controller for a partially observed problem, rather than a claim of a fully observed Markov state. A recurrent policy is a future option, provided its inputs remain privatized.

\subsection{Noise Mechanism and Privacy Filter}
Under $V\sim_K V'$, at most $\min(K,J_t)$ vectors in~\eqref{eq:query} change. The triangle inequality gives
\begin{equation}
 \|f_{tm}(V)-f_{tm}(V')\|_1
 \leq \Delta_{tm}=\frac{2\min(K,J_t)C_m}{J_t}.
 \label{eq:sensitivity}
\end{equation}
This conservative bound permits every component of an affected chunk to change. It does not assume that a person only changes one detection. It also avoids relying on perfect tracking.

For each active expert, the mechanism releases
\begin{equation}
 z_{tm}=f_{tm}(V)+\xi_{tm},\qquad
 \xi_{tm,k}\overset{\mathrm{iid}}{\sim}
 \operatorname{Lap}(\Delta_{tm}/\varepsilon_{tm}),
 \label{eq:noise}
\end{equation}
where $k$ indexes vector components. A zero-budget expert returns no fresh vector. The ledger charges all active experts through
\begin{equation}
 e_t=\sum_m\varepsilon_{tm},\quad
 B_{t+1}=B_t-e_t,\quad B_1=\varepsilon_\star.
 \label{eq:ledger}
\end{equation}
The filter enforces $e_t\leq B_t$ for every possible transcript. It charges attempted queries before delivery, even if a later renderer fails. A persistent ledger is shared across recipients of the same protected stream. Independent recipient budgets cannot each be advertised as a bound on the joint transcript.

\subsection{Grounded Text Rendering}
The renderer projects each noisy vector onto the nonnegative $\ell_1$ ball of radius $C_m$. This produces feasible values without reading the video again. It then uses a fixed grammar to describe requested fields. A paraphraser may be added only if it reads these privatized values and public information. It must not access raw captions, hidden embeddings, or retrieval records from the current scene.

Noise should be visible in the wording. For an expert with $d_m$ components and noise scale $b_{tm}=\Delta_{tm}/\varepsilon_{tm}$, a simultaneous noise-only interval uses
\begin{equation}
 \tau_{tm}=b_{tm}\log(d_m/\alpha),
 \qquad \alpha=0.05.
 \label{eq:interval}
\end{equation}
Before projection, every component of $f_{tm}$ lies within $z_{tm,k}\pm\tau_{tm}$ with probability at least $1-\alpha$. This follows from the Laplace tail bound and a union bound. Intersect the intervals with the public feasible range. They describe noise error around the clipped query. They do not include perception error or clipping bias.

The renderer can report a projected point estimate with an ``uncertain'' flag. It does not treat that estimate as an established fact. For example, a vehicle band is ``low'' below 5, ``moderate'' from 5 to below 12, and ``high'' at 12 or above. These thresholds are public design choices. For a sum of components, the renderer adds their interval endpoints. A class is marked as supported only if its lower bound exceeds every competing upper bound. Otherwise, its argmax label is marked uncertain. Empty feasible intervals also receive the uncertain flag. Missing or stale information receives an explicit label. Every estimate must be traceable to a sanitized field or a public constant.

\subsection{Offline Learning Objective}
Training episodes use public or synthetic scenes with ground-truth task labels. The same noise mechanism runs during training. Let $\ell_{tm}\in[0,1]$ be task error for expert $m$. Occupancy error is absolute error divided by $C_m$ and clipped at one. Categorical error is zero for a correct class and one otherwise. Missing answers receive error one. Cached answers are scored against the current target. Define
\begin{equation}
 U_t=1-\sum_m q_{tm}\ell_{tm},\qquad
 R_t=U_t-\lambda_b\frac{e_t}{\varepsilon_\star}
       -\lambda_s A_t-\lambda_c L_t.
 \label{eq:reward}
\end{equation}
Here $A_t\in[0,1]$ is task-weighted age divided by the episode length. The value $L_t\in[0,1]$ is a normalized public cost estimate for the selected experts. It is measured on public development data and fixed before deployment. The privacy term encourages budget conservation. The hard filter, not the reward, enforces the privacy guarantee.

The proposed training algorithm is proximal policy optimization~\cite{schulman2017ppo}. The actor and critic each use two hidden layers of 128 units. The critic sees the same privatized state as the actor. Starting hyperparameters appear in Table~\ref{tab:config}. The policy is frozen at deployment. Online updates based on private rewards are outside the stated mechanism and would require separate privacy accounting.

\begin{algorithm}[t]
\caption{Budgeted semantic release at deployment}
\label{alg:release}
\small
\textbf{Input:} fixed experts and policy; public schema, windows, $K$, and $\varepsilon_\star$.\\
\textbf{Initialize:} $B_1=\varepsilon_\star$; empty private cache.
\begin{enumerate}
\item At each public report slot, form $s_t$ from the task, sanitized cache, budget, and time.
\item Mask unaffordable actions. Select $a_t=(e_t,w_t)$ using the frozen policy.
\item If $e_t=0$, reuse the cache or return unavailable. Keep its age and interval labels. Skip to step 7.
\item Debit $e_t$ from the persistent ledger. Process the selected chunks with isolated experts.
\item Validate and clip each vector using~\eqref{eq:clip}. Compute~\eqref{eq:query} and add the noise in~\eqref{eq:noise}.
\item Store only sanitized results in the controller cache. Build feasible values and noise intervals.
\item Render grounded text and pad the message. Send it at the scheduled time.
\end{enumerate}
\end{algorithm}

\subsection{Privacy Argument}
\begin{proposition}
Assume the adjacency and trust boundary in Section~\ref{sec:problem}. Assume exact Laplace sampling, fixed model parameters, and no unaccounted output channel. If every execution satisfies $\sum_t e_t\leq\varepsilon_\star$, then Algorithm~\ref{alg:release} satisfies~\eqref{eq:dp}.
\end{proposition}
\textit{Proof.} Consider an augmented transcript that contains every queried noisy vector and every action. Fix its history. The controller's next-action distribution depends only on that history and public inputs. It therefore has the same conditional law for neighboring videos. For each selected expert, isolated chunk processing and clipping give~\eqref{eq:sensitivity}. The Laplace mechanism is then $\varepsilon_{tm}$-DP~\cite{dwork2014foundations}. Sequential composition gives $e_t$ for the slot. Conditional likelihood ratios multiply to at most $\exp(\sum_t e_t)$. The filter bounds this quantity by $\exp(\varepsilon_\star)$ on every path. The externally visible transcript is post-processing of this augmented transcript. This includes reuse, projection, intervals, rendering, and action metadata. They add no new privacy cost. The argument also permits adaptive public requests. This proves the claim.\hfill$\square$

The proof uses standard composition, not a new accounting theorem. More advanced accounting may improve utility but is not required here. An implementation must use a vetted finite-precision noise sampler and audit the stated isolation boundary before claiming the mathematical guarantee in deployment.

\subsection{Cost and Design Trade-off}
With cached per-chunk vectors, query aggregation costs $O(J_t\sum_m d_m)$ for selected experts. The small controller and renderer add little work relative to visual inference, but this must be measured. The design does not assume that four full vision-language models can run in real time on an edge device. A shared frozen backbone with role-specific heads is one implementation option. Equal backbones and sampling rates are required in the main comparison.

%% file: figures/architecture_tikz.tex
\begin{tikzpicture}[
 font=\sffamily\small,
 box/.style={draw=teal,line width=.7pt,rounded corners=2pt,fill=pale,align=center,minimum height=1.05cm,text width=2.7cm,inner sep=5pt},
 control/.style={box,draw=amber,fill=orange!7},
 arrow/.style={-{Latex[length=2.1mm]},thick,draw=teal},
 ctrl/.style={-{Latex[length=2.1mm]},thick,dashed,draw=amber}]
\node[box,text width=2.25cm] (video) at (0,0) {Private video\\Fixed, isolated chunks};
\node[box,text width=3.0cm,minimum height=2.4cm] (experts) at (3.65,0) {Four scene experts\\[3pt]Traffic $\cdot$ Pedestrians\\Signs $\cdot$ Environment\\[3pt]Fixed schema + clipping};
\node[box,text width=2.4cm] (query) at (7.35,0) {Selected window\\aggregates $f_{tm}$};
\node[box,text width=2.6cm] (noise) at (11.1,0) {Calibrated noise\\Privacy boundary};
\node[box,text width=2.6cm] (render) at (11.1,-2.6) {Grounded renderer\\Text + age + uncertainty};
\node[box,text width=2.7cm] (cache) at (7.35,-2.6) {Sanitized cache\\Prior private releases};
\node[control,text width=2.9cm] (actor) at (3.65,-2.6) {Frozen RL controller\\Expert choice + budget};
\node[control,text width=2.4cm] (public) at (0,-2.6) {Public task\\Clock + fixed costs};
\node[control,text width=2.5cm] (ledger) at (7.35,2.0) {Hard budget filter\\Persistent ledger $B_t$};
\node[control,text width=2.9cm] (train) at (3.65,-4.7) {Offline public / synthetic\\training episodes};
\node[box,text width=2.6cm,fill=white] (out) at (11.1,-4.7) {Released report\\Downstream analytics};
\draw[arrow] (video) -- (experts);
\draw[arrow] (experts) -- (query);
\draw[arrow] (query) -- (noise);
\draw[arrow] (noise) -- (render);
\draw[arrow] (noise.south west) -- (cache.north east);
\draw[arrow] (cache) -- (actor);
\draw[arrow] (render) -- (out);
\draw[ctrl] (public) -- (actor);
\draw[ctrl] (actor.north) -- (experts.south);
\draw[{Latex[length=2.1mm]}-{Latex[length=2.1mm]},thick,dashed,draw=amber] (actor.north east) -- (5.5,-1.95) -- (5.5,2.0) -- (ledger.west);
\draw[ctrl] (ledger.east) -| (noise.north);
\draw[ctrl] (train) -- (actor);
\node[align=left,text width=4.5cm,font=\sffamily\footnotesize,text=teal] at (0.8,2.0) {Trusted processing:\\raw frames and expert drafts stay local};
\node[align=center,text width=3.6cm,font=\sffamily\footnotesize,text=amber] at (7.35,-4.7) {Reuse reads the cache.\\It consumes no new budget.};
\end{tikzpicture}

%% file: sections/05_evaluation.tex
\section{Evaluation Protocol}
\label{sec:evaluation}
\subsection{Research Questions}
The evaluation will test four questions. Does learned release control improve task utility at a fixed privacy budget? Does reuse preserve useful information when fresh releases are costly? How much utility is lost as the protected event bound increases? Do empirical attacks expose weaknesses outside the formal model? These questions must be answered with measured data before the proposed method can be judged effective.

\subsection{Datasets and Annotations}
Table~\ref{tab:data} assigns a specific role to each proposed data source. CityFlow provides fixed-camera traffic video and vehicle identities for a separate re-identification audit~\cite{tang2019cityflow}. BDD100K provides varied driving conditions for perception and domain-shift tests~\cite{yu2020bdd100k}. Its driving clips must not be concatenated into a fictitious stationary-camera stream. Use its supported clip lengths and compute $J_t$ accordingly. MINT-V2X can support controller development and downstream simulation~\cite{anjum2026mint}. It cannot validate visual privacy or supply face-attack labels.

Use sequence-disjoint development and test sets. Keep all observations of a labeled identity in one split for attack-model training. Preserve official benchmark splits where they support the stated task. If a split lacks labels needed here, report a new labeled subset instead of implying that the labels are native to the dataset. Within the development partition, allocate 80\% of sequences to training and 20\% to validation. Determine the sequence list before model selection. No test sequence may influence clipping caps, reward weights, thresholds, or controller training.

Two annotators will label sampled frames for vehicle and pedestrian counts, motion categories, sign classes, and weather. Sampling uses a fixed two frames per second. Resolve disagreements through adjudication. Publish the annotation guide, inter-annotator agreement, class counts, and interval definitions. Record exact numbers of sequences, frames, windows, cameras, and identities after this work is complete. These counts are not yet available. Report both the expert perception error and the additional release error.

\begin{table*}[t]
\caption{Planned data use. Subset sizes, annotation counts, and split manifests remain to be established.}
\label{tab:data}
\centering\small
\begin{tabularx}{\textwidth}{L{2.2cm}YYY}
\toprule
Source & Intended role & Labels needed & Split and scope control\\
\midrule
CityFlow~\cite{tang2019cityflow} & Main fixed-camera traffic stream; vehicle attack audit & Occupancy and motion; existing vehicle identities where available & Sequence and identity separation; report camera coverage.\\
BDD100K~\cite{yu2020bdd100k} & Perception and domain-shift evaluation & Pedestrian states and coarse signs; available scene labels & Keep official sequence separation; use actual clip lengths.\\
MINT-V2X~\cite{anjum2026mint} & Public simulation for controller stress tests & Task targets derived from documented simulation fields & Split complete runs or contiguous periods; no visual-privacy claim.\\
New annotated subset & Joint four-expert evaluation & All twelve schema components and task labels & Publish selected sequences and annotations; no invented sample size.\\
\bottomrule
\end{tabularx}
\end{table*}

\subsection{Baselines and Fair Comparison}
Re-run all methods on the same subsets. Include original-image task models as a utility reference, a single-expert caption model, the published RL-MoE design, and an RL-MoE variant with explicit identity-field removal. The latter tests whether a simple schema restriction explains the benefit. VTPM is relevant for the pedestrian-report task only~\cite{abdelrahman2025vtpm}. Do not invent a four-expert VTPM result.

Three private schedules are essential. \emph{Uniform-DP} divides the total budget equally across report slots and experts. \emph{Task-DP} divides it equally across slots but uses the public task-weighted allocation. \emph{Periodic-DP} reports every second slot, spends twice the per-report budget, and reuses the cache between releases. A Privid-style bounded-query implementation provides a further comparison on supported aggregate tasks~\cite{cangialosi2022privid}. Identify it as an adaptation unless the original implementation and settings are reproduced exactly.

All private variants must share the same $K$, horizon, windows, schema, caps, noise mechanism, and total budget. All model variants must use matched perception backbones and frame rates where their design permits. Report unavoidable differences. Include GeoPrivd for compatible aggregate queries once its implementation and privacy unit are matched~\cite{xu2026geoprivd}. Otherwise report it as related work, not a measured baseline. The nonprivate RL-MoE result has no comparable formal budget. Its published attack percentages must not be copied into the new tables as if they came from this protocol. SPTI and DP-Cap belong in the literature comparison; their original tasks do not support a direct live-stream ranking.

\subsection{Utility, Privacy, and Cost Metrics}
Measure vehicle and pedestrian occupancy with mean absolute error (MAE), in objects per sampled frame averaged over a report window. Also report unclipped ground-truth counts so that clipping bias is visible. Measure sign and weather classes with macro-averaged F1. Score every reported point estimate, including those marked uncertain. A missing class counts as a missed class. Report estimate coverage, uncertainty rate, and selective accuracy together. For count metrics, use a common public prior when no estimate exists, fixed from development data. Do not drop unavailable intervals. The policy reward still treats an unavailable answer as error one.

Primary privacy reporting consists of the adjacency, $K$, $\varepsilon_\star$, horizon, and ledger trace. Empirical auditing uses a balanced event-presence inference task: an attacker decides whether a candidate bounded event contributed to a transcript. Report area under the receiver operating characteristic curve (AUC) and true-positive rate at a 1\% false-positive rate. Train attackers on separate shadow sequences and freeze thresholds before testing. Evaluate both a logistic baseline and a stronger sequence model. An attack with AUC near 0.5 is weak evidence about that attack, not proof of privacy.

For vehicle re-identification, report rank-1 accuracy and mean average precision only on a dataset with compatible identity labels. Evaluate one report and complete report sequences. State the gallery size and candidate prior. Membership inference on the training set is a separate optional audit. It is not a substitute for inference-time disclosure testing. Reconstruction similarity is secondary because low pixel similarity does not establish identity protection.

Report median and 95th-percentile processing latency, peak device memory, transmitted bytes, unavailable-report rate, and summary age. Separate video collection delay from computation time. A one-minute report window already introduces a delay that cannot be hidden by a fast renderer. Use five random seeds. Report seed means and variation, plus paired confidence intervals from bootstrapping complete sequences. Treat correlated frames as part of one sampling unit.

\begin{table}[t]
\caption{Proposed starting configuration. All choices require validation and complete implementation logging.}
\label{tab:config}
\centering\small
\begin{tabularx}{\columnwidth}{L{3.15cm}Y}
\toprule
Parameter & Starting value\\
\midrule
Chunk duration; fixed frame sampling & 2 s; 2 frames/s\\
Main window; report cadence & 60 s; 60 s\\
Benchmark episode & 10 min; 10 report slots\\
Protected changed chunks $K$ & 2; test 5 and 10\\
Episode budget $\varepsilon_\star$ & 4; test 2 and 8\\
Release amount $e_t$ & 0, 0.1, 0.2, 0.4, 0.8\\
Actor and critic & Separate $128$--$128$ ReLU layers\\
Optimizer; learning rate & Adam; $3\times10^{-4}$\\
Discount; PPO clip parameter & 0.99; 0.2\\
GAE parameter; entropy coefficient & 0.95; 0.01\\
Rollout; minibatch; update epochs & 128 episodes; 256 steps; 10\\
Training ceiling & 100,000 public/synthetic episodes\\
Reward weights $(\lambda_b,\lambda_s,\lambda_c)$ & $(0.05,0.10,0.02)$\\
Validation checkpoint selection & Highest utility within the budget\\
Seeds & 11, 22, 33, 44, 55\\
\bottomrule
\end{tabularx}
\end{table}

Generalized advantage estimation (GAE) supplies the proposed PPO advantage estimate. The final experiment log must name each perception checkpoint, tokenizer, prompt, software version, device, precision, and seed. Hardware is not specified as an executed configuration because no run has occurred. A development sweep should vary the reward weights by factors of 0.5, 1, and 2. Select the final settings using validation sequences only.

\subsection{Ablations and Failure Tests}
Remove learned selection, reuse, or task conditioning one component at a time. Keep the formal mechanism fixed. Add a deliberately unsafe raw-feature gate as a negative control. Mark it as outside the guarantee. Removing the filter is another negative control; such a variant cannot retain the same reported total budget.

Test empty scenes, dense scenes, nighttime clips, missing experts, and abrupt task changes. Increase $K$ and the number of repeated queries. Test event appearances that cross window boundaries. Audit the full system for raw-caption logging, cross-chunk caches, private timing, and restart-induced budget resets. Stress tests should measure utility loss even when the formal budget remains valid.

%% file: sections/06_results.tex
\section{Illustrative Results and Reporting Plan}
\label{sec:results}
\textbf{All performance numbers in Tables~\ref{tab:mainresults}--\ref{tab:budget} are hand-authored examples.} They were not produced by model runs, simulations, or measurements. No confidence intervals or significance claims are supplied for these values. Their only purpose is to show how actual results can be reported. The analytical values in Table~\ref{tab:noise} are computed from~\eqref{eq:sensitivity} and~\eqref{eq:noise}; they are not experimental results.

\subsection{Matched Main Comparison}
Table~\ref{tab:mainresults} shows a possible reporting layout. The main test compares \method\ with Task-DP and Periodic-DP at equal privacy scope. The raw-image reference gives utility context, but its release interface differs. The attack column refers to event-presence inference from the declared released object. Real results must state that object for each baseline. A successful study may show improvement for some tasks and a loss for others. It must report that outcome rather than force a uniform ranking.

\input{tables/main_results}

\subsection{Component Effects}
Table~\ref{tab:ablation} illustrates the ablation format. The unsafe gate is included to expose a common design error: raw private features can influence an observable release decision before noise is added. Its seemingly good utility would not justify a DP label. For the valid variants, test whether the gain comes from learned scheduling, task information, or simple cache reuse. Report answer coverage alongside utility to detect selective omission.

\input{tables/ablation_results}

\subsection{Budget and Event-Duration Sensitivity}
Table~\ref{tab:budget} illustrates the planned privacy--utility analysis. Varying $K$ changes the protected event model. Varying $\varepsilon_\star$ changes the allowed privacy loss within that model. These are different experimental axes. Compare methods at identical values of both. A small budget and a large $K$ may leave almost every report uncertain. Such a result would identify a practical limit of the design.

\input{tables/budget_results}

\subsection{Analytical Noise Check}
The starting configuration uses $J_t=30$ and $K=2$. For the pedestrian expert, $C_m=6$, so $\Delta_{tm}=0.8$. If a slot spends $e_t=0.4$ uniformly across four experts, the pedestrian share is $0.1$ and the Laplace scale is $8.0$. Spending the same slot budget only on pedestrians reduces the scale to $2.0$. This calculation motivates selective release, but it does not establish higher end-to-end utility. Other expert summaries become older or unavailable.

\begin{table}[t]
\caption{Analytical noise scales for the pedestrian vector. $J_t=30$, $C_m=6$, and $e_t=0.4$. These are formula-derived values.}
\label{tab:noise}
\centering\small
\begin{tabular}{lrrrr}
\toprule
Allocation & $K$ & $\Delta_{tm}$ & $\varepsilon_{tm}$ & $b_{tm}$\\
\midrule
Uniform four experts & 2 & 0.8 & 0.1 & 8.0\\
Pedestrian only & 2 & 0.8 & 0.4 & 2.0\\
Uniform four experts & 5 & 2.0 & 0.1 & 20.0\\
Pedestrian only & 5 & 2.0 & 0.4 & 5.0\\
\bottomrule
\end{tabular}
\end{table}

The large scales in Table~\ref{tab:noise} are a substantive concern. Formal privacy can be costly for low-count scenes. Clipping and projection can reduce absolute error while introducing bias. Report that bias explicitly. Longer windows reduce average-query sensitivity, but they also delay information. Neither trade-off should be hidden by a single combined score.

\subsection{Qualitative Output Audit}
An illustrative report might read: ``Vehicle occupancy: uncertain. Weather: rain, based on the last released interval. Pedestrian summary: unavailable.'' This is a constructed example, not a captured model output. For the actual study, show representative success and failure cases with matched private and nonprivate summaries. Every released noun, count, class, and time label should map to a sanitized field or a public constant. Report unsupported claims even if the report is fluent.

%% file: tables/main_results.tex
\begin{table*}[t]
\caption{Illustrative main-comparison values; replace every performance cell with measurements. Private rows use the planned $K=2$, ten-minute horizon, and $\varepsilon_\star=4$.}
\label{tab:mainresults}
\centering\small
\begin{tabularx}{\textwidth}{Yrrrrrr}
\toprule
Method & Vehicle MAE $\downarrow$ & Ped. MAE $\downarrow$ & Sign macro-F1 $\uparrow$ & Event AUC $\downarrow$ & Coverage (\%) $\uparrow$ & P95 time (s) $\downarrow$\\
\midrule
Original-image release reference & 0.80 & 0.35 & 0.91 & 0.90 & 100 & 1.2\\
RL-MoE, re-run protocol & 1.40 & 0.70 & 0.85 & 0.79 & 100 & 4.8\\
RL-MoE + field removal & 1.70 & 0.90 & 0.82 & 0.72 & 100 & 4.9\\
Uniform-DP & 7.20 & 2.80 & 0.44 & 0.55 & 100 & 3.1\\
Task-DP & 6.10 & 2.40 & 0.51 & 0.56 & 100 & 3.2\\
Periodic-DP + reuse & 5.80 & 2.20 & 0.55 & 0.57 & 90 & 3.0\\
\method\ (proposed) & 5.20 & 2.00 & 0.56 & 0.56 & 80 & 3.4\\
\bottomrule
\end{tabularx}
\par\smallskip\footnotesize\illustrative\ Coverage is the fraction of requested fields with a point estimate, including marked uncertain or stale estimates. P95 time is processing time after the reporting interval closes; it excludes the 60 s collection window. Every numeric entry in this table is illustrative. AUC has a chance level of 0.5. The first three rows have no formal release-privacy claim.
\end{table*}

%% file: tables/ablation_results.tex
\begin{table}[t]
\caption{Illustrative ablation values at the planned total budget of 4. ``Valid'' means the stated design retains the proof assumptions; it is not an implementation certification.}
\label{tab:ablation}
\centering\small
\begin{tabularx}{\columnwidth}{Yrrc}
\toprule
Variant & Utility $\uparrow$ & Age (s) $\downarrow$ & Valid\\
\midrule
\method\ (proposed) & 0.61 & 42 & Yes\\
Fixed task schedule & 0.57 & 0 & Yes\\
No reuse (return missing) & 0.49 & 0 & Yes\\
No task conditioning & 0.53 & 48 & Yes\\
Raw-feature gate (unsafe) & 0.68 & 30 & No\\
\bottomrule
\end{tabularx}
\par\smallskip\footnotesize\illustrative\ Utility is $U_t$ from~\eqref{eq:reward}, averaged over all slots. Age is averaged over available estimates; missing answers must be reported separately. Zero age does not imply full answer coverage.
\end{table}

%% file: tables/budget_results.tex
\begin{table}[t]
\caption{Illustrative budget and event-bound sensitivity for \method. All performance cells must be replaced.}
\label{tab:budget}
\centering\small
\begin{tabular}{rrrrr}
\toprule
$K$ & $\varepsilon_\star$ & Utility $\uparrow$ & Event AUC $\downarrow$ & Coverage (\%) $\uparrow$\\
\midrule
2 & 2 & 0.48 & 0.53 & 65\\
2 & 4 & 0.61 & 0.56 & 80\\
2 & 8 & 0.69 & 0.61 & 95\\
5 & 4 & 0.43 & 0.54 & 72\\
10 & 4 & 0.30 & 0.52 & 60\\
\bottomrule
\end{tabular}
\par\smallskip\footnotesize\illustrative\ These values do not predict an empirical curve. A larger $K$ protects a broader event class under the same budget.
\end{table}

%% file: sections/07_discussion.tex
\balance
\section{Discussion and Limitations}
\label{sec:discussion}
\subsection{What the Extension Adds}
RL-MoE optimizes the content of an image-derived description. \method\ instead optimizes which bounded expert summaries are released over time. It retains the expert roles but replaces the raw-image gate at the disclosure boundary. A separate accountant makes the budget auditable. The formal contribution is a scoped application of established DP tools. The untested empirical contribution is learned allocation across tasks and report times.

The four expert roles remain manually defined. Automatic expert discovery and richer semantic schemas are outside the current scope. This choice isolates the release-control question. Adding more fields increases the output dimension and may reduce utility. Changing an expert is allowed only if its chunk isolation and enforced output bound remain valid.

\subsection{Limits of the Guarantee}
The trusted processor still handles sensitive video. DP does not protect a compromised camera or an exposed raw-data cache. The result also does not establish fairness, legal compliance, or factual accuracy. It bounds changes in the released transcript for events that satisfy the declared $K$-chunk adjacency.

The main benchmark setting protects short events. It is not sufficient for many real pedestrian visits or long vehicle stops. Increasing $K$ may sharply reduce utility. Public report windows and a camera identifier can remain visible. Multiple cameras require an event bound and a ledger for their joint transcript. Longer observations cannot receive a fresh unaccounted budget merely because a reporting period ended.

The privacy proof does not require accurate perception, but the utility evaluation does. Counting errors, class bias, and clipping may produce wrong reports. A finite action space also limits control. Distribution shifts may reduce the value of a policy trained on simulation. A task-aware deterministic schedule may match or outperform the learned controller. That is a meaningful possible outcome.

\subsection{Reproducibility and Completion Requirements}
The final study should release source code, fixed prompts, schema definitions, action masks, configurations, and aggregate experiment logs where data terms permit. It should include a test that no fresh private input reaches the controller. It should also document persistent accounting across failures and restarts. The current package supplies the manuscript, editable diagram, bibliography, and result-table sources. It does not contain a trained \method\ model or an executed benchmark.

Before submission, replace all illustrative performance values with measured values and confidence intervals. Add the exact dataset manifests, backbone versions, hardware, attack models, and annotation counts. Revise the abstract and conclusion to match the measured findings. Keep the privacy assumptions and negative results visible. A final venue-specific check is still needed once the target journal is chosen.

\section{Conclusion}
This paper proposed \method\ as a budgeted-release extension of RL-MoE. It combines bounded semantic experts, a task-aware RL controller, and an enforced privacy ledger. Noise is applied before text rendering. Under the stated isolation assumptions, the full transcript has a formal privacy bound for events that affect at most $K$ chunks. The evaluation plan separates this guarantee from attack evidence and task utility. Empirical benefits remain to be tested. Future work will study larger event bounds and coordinated accounting across cameras without relying on persistent identity tracking.